\documentclass[pra, nofootinbib]{revtex4}
\usepackage{tabularx,amsmath,boxedminipage,graphicx,amssymb}
\usepackage{braket, dsfont}
\usepackage{setspace}
\usepackage{qcircuit}
\usepackage{titlesec}
\usepackage{amsthm}
\usepackage{color}
\usepackage{url}
\usepackage[final]{hyperref} 
\hypersetup{
	colorlinks=true,       
	linkcolor=blue,          
	citecolor=blue,        
	filecolor=magenta,      
	urlcolor=black         
}
\newcommand{\beq}{\begin{equation}}
\newcommand{\eeq}{\end{equation}}
\definecolor{Pr}{rgb}{0.4,0.3,0.9}

\definecolor{JM}{RGB}{4,116,149}

\newtheorem{obs}{Observation}

\begin{document}
\title{Exact gate decompositions for photonic quantum computing}
\author{Timjan Kalajdzievski}
\author{Juan Miguel Arrazola}

\affiliation{Xanadu, 372 Richmond St W, Toronto, M5V 2L7, Canada}

\begin{abstract}
We propose a method for decomposing continuous-variable operations into a universal gate set, without the use of any approximations. We fully characterize a set of transformations admitting exact decompositions and describe a process for obtaining them systematically. Gates admitting these decompositions can be synthesized exactly, using circuits that are several orders of magnitude smaller than those achievable with previous methods. Our method relies on strategically using unitary conjugation and a lemma to the Baker-Campbell-Hausdorff formula to derive new exact decompositions from previously known ones, leading to exact decompositions for a large class of gates. We demonstrate the wide applicability of these exact gate decompositions by identifying several quantum algorithms and simulations of bosonic systems that can be implemented with higher precision and shorter circuit depths using our techniques. 
\end{abstract}

 \maketitle

\section{Introduction}

A quantum algorithm is usually specified by a sequence of high-level unitary transformations \cite{ChrisAlgo2017, kalajdzievski2018continuous, CVHHL, Loock2013, kitaev1997quantum, SolovayReview, amy2013meet}. A physical quantum computer, on the other hand, is only capable of performing a small set of elementary gates. The challenge of programming a quantum computer is to find combinations of these elementary gates that can reproduce a desired algorithm. It is known that specific sets of logical gates exist such that any arbitrary unitary operation can be expressed as a finite product of gates from the set, to any desired precision \cite{lloyd1995almost, divincenzo1995two, barenco1995elementary, SethUniversal, SethCV, bravyi2005universal}. Given their ability to reproduce any desired transformation, these are referred to as universal gates sets. Programming a quantum computer to perform a desired algorithm thus requires a method to decompose high-level unitaries in terms of universal gate sets. Ideally, a decomposition method will reproduce the algorithm with high precision while requiring as few gates as possible. \\

Gate decompositions in the qubit model of quantum computing have been well studied. For example, the Solovay-Kitaev theorem \cite{kitaev1997quantum, SolovayReview} states that if a set of qubit gates generates a dense subset of $SU(2)$, then it can approximate any $SU(2)$ unitary using a number of gates that is logarithmic in the precision. This means that any single-qubit operation can be approximated to high precision using short circuits. These results have been strengthened to even more efficient decompositions for single-qubit operations \cite{kliuchnikov2013asymptotically, kliuchnikov2014asymptotically, kliuchnikov2015framework, bocharov2015efficient} and general multi-qubit operations \cite{amy2013meet}.\\ 

In the continuous-variable (CV) model of quantum computing, registers are infinite-dimensional quantum systems -- namely quantum harmonic oscillators -- and the logic gates are unitaries acting on the infinite-dimensional Hilbert space \cite{SethCV, Menicucci2006, Gu2009, ChrisOverview, ChrisAlgo2017, kalajdzievski2018continuous, CVHHL, Loock2013}. This presents unique challenges for the task of decomposing arbitrary operations in CV photonic quantum computers, where comparatively less progress has been made thus far. Ref.~\cite{SethCV} introduced the notion of universality in CV quantum systems based on the commutator algebra of quadrature operators. Following this, Ref.~\cite{CommutatorApprox} presented the first systematic approach for decomposing arbitrary CV transformations, while Refs.~\cite{ChrisAlgo2017, kalajdzievski2018continuous, CVHHL, Loock2013} deal with decompositions for specific tasks. All these methods are approximate in the sense that the resulting sequence of gates from the universal set only implements the desired unitary up to a certain error, which can be decreased arbitrarily by employing longer circuits \cite{Suzuki2005, CommutatorApprox, Wiebe2010}. However, this can lead to very large circuit depths even if a modest precision is desired. \\
    
   In this work, we introduce a method for decomposing a wide class of quantum gates without the use of any approximations. Exact decompositions are known for a few specific cases \cite{ChrisAlgo2017, CommutatorApprox, CVHHL}, but it is not well understood what transformations allow exact decompositions nor how they can be derived. We remedy this by characterizing a set of gates admitting exact decompositions and by describing a process for obtaining them systematically. This leads to circuits that are several orders of magnitude shorter than for previously known techniques, with no errors due to approximations. For example, decomposing the gate $e^{i\hat{X}^4}$ with the commutator approximation method requires approximately $1.8 \times 10^4$ gates with precision $10^{-3}$, while using our method only 29 gates are needed and the gate is decomposed exactly. We demonstrate the applicability of these operations by compiling a table of CV algorithms and simulations of bosonic systems for which exact decompositions can be employed. \\

   The remainder of this work is structured as follows. We begin by introducing the basic concepts and commonly used tools in CV gate decomposition. 
We then outline our exact decomposition method and discuss how it can be used in Sec.~\ref{ExactMethods}. Finally, in Sec.~\ref{Comparing}, we compare our method to previous techniques and examine some specific examples where our method may be applied. Sec.~\ref{Discussion} offers a brief discussion and some insights into open questions.

\section{GATE DECOMPOSITIONS} \label{Background}
   In the CV model of quantum computing, each register is a quantum harmonic oscillator with corresponding creation and annihilation operators $\hat{a}_j$ and $\hat{a}^{\dagger}_j$, where the subscript refers to the mode they act upon. For definiteness, we henceforth assume that these registers are modes of the quantized electromagnetic field. The annihilation and creation operators satisfy the bosonic commutation relations $[\hat{a}_j , \hat{a}^{\dagger}_j]=1$, and $[\hat{a}_j, \hat{a}_k]=[\hat{a}^{\dagger}_j, \hat{a}^{\dagger}_k]=0$ for $j\neq k$. An equivalent operator description of a bosonic system uses the quadrature field operators $\hat{X}$ and $\hat{P}$, which are related to the annihilation and creation operators as
\begin{align}
\hat{X}_{j} = \frac{1}{2}\left(\hat{a}^{\dagger}_j+ \hat{a}_j\right), \\
\hat{P}_{j} = \frac{i}{2}\left(\hat{a}^{\dagger}_j - \hat{a}_j\right), 
\end{align} 
with commutator $[\hat{X}_j,\hat{P}_j] = \frac{i}{2}$.\\

   A universal gate set is a collection of gates such that any arbitrary unitary operation can be expressed as a finite series of gates from the universal set, to any chosen approximation. We focus on the universal set specified by the gates 
\beq
\{e^{i\frac{\pi}{2}(\hat{X}_j^2+\hat{P}_j^2)},\, e^{it_1\hat{X}_j},\, e^{i t_2 \hat{X}_j^2},\, e^{it_3 \hat{X}_j^3},\, e^{i\tau \hat{X}_j\hat{X}_k }\},
\eeq
where $t_{1}$, $t_{2}$, $t_{3}$, and $\tau$ are real parameters. This particular universal set is chosen for mathematical convenience in our method. The gate $e^{i\tau \hat{X}_1\hat{X}_2 }$ allows for decompositions of multiple modes, while the Fourier transform gate $\hat{F}=e^{i\frac{\pi}{2}(\hat{X}^2+\hat{P}^2)}$ has the effect of mapping between the quadrature operators:
\begin{align}\label{Fourier1}
\hat{F}^{\dagger}\hat{X}\hat{F}&=-\hat{P},\\ \label{Fourier2}
\hat{F}^{\dagger}\hat{P}\hat{F}&=\hat{X}.
\end{align}

For convenience, we express an arbitrary unitary as $U = e^{it\hat{H}}$ with $\hat{H}=\sum_{j=1}^N \hat{H}_{j}$ a Hermitian operator. When decomposing gates into a universal set, it is often necessary to express this sum of operators in the exponent as a product of exponential operators. More specifically, for $\hat{H} = \hat{A} + \hat{B}$ where $\hat{A}$ and $\hat{B}$ are Hermitian operators, the Zassenhaus formula \cite{magnus1954exponential} states that
\beq \label{splitting}
e^{it(\hat{A}+\hat{B})} = e^{it\hat{A}}e^{it\hat{B}}e^{\frac{t^2}{2}[\hat{A},\hat{B}]}e^{\frac{-it^3}{6}\left(2[\hat{B},[\hat{A},\hat{B}]]+[\hat{A},[\hat{A},\hat{B}]]\right)}\cdots
\eeq
In the trivial case where $[\hat{A},\hat{B}]=0$ the product ends immediately after the first two operations. In general, however, it is possible that this product never terminates, resulting in a decomposition that is no longer finite. In this case, it is possible to truncate the product at a designated stage in the expansion and neglect the proceeding commutators. This strategy is referred to as a Trotter-Suzuki approximation \cite{Suzuki1976}, which can be stated in the general case as
\beq\label{Eq: trotter-suzuki}
e^{it\hat{H}}=\prod_{j=1}^N\left(e^{i\frac{t}{K}\hat{H}_j}\right)^K+O(t^2/K),
\eeq 
where $\hat{H} = \sum_{j=1}^N \hat{H}_{j}$. This approximation requires $K=O(1/\varepsilon)$ gates to achieve precision $\varepsilon$ for fixed $t$. 

In general, the unitaries of the form $e^{it\hat{H}_j}$ are not part of the universal set, so the task remains to decompose them. One way to achieve this is via the commutator approximation method detailed in Ref.~\cite{CommutatorApprox}. This technique expresses sums and products of the quadrature operators in terms of commutators and then approximates the exponentials of these commutators as repeated products of their arguments. More specifically, given two Hermitian operators $\hat{A}$ and $\hat{B}$, it holds that \cite{childs2017toward}
\begin{align}\label{Eq: comm_approx}
e^{t^2[\hat{A},\hat{B}]}&= \left(e^{i\frac{t}{K}\hat{B}}e^{i\frac{t}{K}\hat{A}}e^{-i\frac{t}{K}\hat{B}}e^{-i\frac{t}{K}\hat{A}}\right)^{K^2}+O(t^4/K).
\end{align}
For fixed $t$, $K=O(1/\varepsilon)$ gates are required to achieve an error of $\varepsilon$ in the approximation, but the resulting circuit will have a depth of $O(1/\varepsilon^2)$. This means that very large circuits are required for even a modest precision. 
To illustrate the use of the commutator approximation technique, consider an example where we wish to decompose the operator $e^{it\left(\hat{X}^2\hat{P} + \hat{P}\hat{X}^2\right)}$. First, using the equality $\hat{X}^2\hat{P} + \hat{P}\hat{X}^2=\frac{2}{3}[\hat{X}^3,\hat{P}^2]$ from Ref.~\cite{CommutatorApprox}, we have
\beq
e^{it\left(\hat{X}^2\hat{P} + \hat{P}\hat{X}^2\right)} = e^{\frac{2t}{3}[\hat{X}^3,\hat{P}^2]}.
\eeq
Using Eq.~(\ref{Eq: comm_approx}) with $\hat{A}=\hat{X}^3$ and $\hat{B}=\hat{P}^2$ leads to
\begin{align} \label{commexample}
e^{\frac{2t}{3}[\hat{X}^3,\hat{P}^2]} = & \left(e^{i\frac{\sqrt{\frac{2t}{3}}}{K}\hat{P}^2}e^{i\frac{\sqrt{\frac{2t}{3}}}{K}\hat{X}^3}e^{-i\frac{\sqrt{\frac{2t}{3}}}{K}\hat{P}^2}e^{-i\frac{\sqrt{\frac{2t}{3}}}{K}\hat{X}^3}\right)^{K^2} + O\left[\left(\frac{2t}{3}\right)^2/K\right].
\end{align}
Each of the gates on the right-hand side are contained within the universal set up to Fourier transforms, but in order to obtain a precision of $O(1/K)$, the product must be repeated $O(K^2)$ times. For instance, for $t=1$, if the goal is to impose a precision of $10^{-3}$, the product of four gates needs to be repeated approximately $10^5$ times.\\

In fact, Ref.~\cite{sparrow2018simulating} examines the experimental error of implementing a sequence of gates on a qubit quantum computer. The results show that as the number of gates is increased, the accumulated physical implementation error eventually supersedes the precision gain from the repetitions. Thus, at some point, more repetitions do not lead to lower errors. This problem remains on a CV quantum computer and further study is required to determine the optimal trade-off between physical error in implementation and precision error in the decomposition. However, if it is possible to find an exact decomposition, then there is no longer any need for this trade-off since the decomposition is fully precise.\\

In the literature on CV decompositions there are specific examples where the commutator approximation and even sometimes Trotter-Suzuki can be bypassed \cite{ChrisAlgo2017, CommutatorApprox, CVHHL}. These cases are desirable, but no general framework has been proposed to characterize the set of gates admitting exact decompositions. In the following section, we detail such a general method for performing exact decompositions. 

\section{METHOD FOR EXACT DECOMPOSITIONS} \label{ExactMethods}

We describe a method to decompose multi-mode gates $e^{it\hat{H}}$, where the operator $\hat{H}$ is of the form
\begin{align}\label{Eq: class_of_H}
\hat{H}= \left(\prod_{j=1}^{N-1} \hat{X}_{j}\right)\hat{X}_N^n
\end{align}
for $n$ a positive integer, as well as single-mode gates $e^{it\hat{H}}$ with
\beq
\hat{H}=\hat{X}^{N}.
\eeq
The label of the modes in Eq.~\eqref{Eq: class_of_H} is arbitrary: the method works for any product where at most one operator has an exponent $n>1$. In both cases we require that $N$ is divisible by either 2 or 3, and in the multi-mode case, the product $nN$ must also be divisible by 2 or 3. These gates can be extended to include momentum quadrature operators $\hat{P}_j$ by Fourier transforms acting on individual modes. As we discuss later in the paper, this set of gates for which exact decompositions can be obtained encompasses a large class of operators arising in several CV quantum algorithms and simulations of bosonic systems.\\ 

The method relies on strategically employing: (i) unitary conjugation
\begin{equation}\label{Eq:Unit_conj}
Ue^{it \hat{H}} U^{\dagger}= e^{it U\hat{H}U^\dagger},
\end{equation}
(ii) a lemma to the Baker-Campbell-Hausdorff (BCH) formula
\begin{align}\label{Eq:BCH}
e^A B e^{-A} = B + [A,B] +& \frac{1}{2!}[A,[A,B]] + \cdots,
\end{align}
and (iii), the identity
\begin{align} \label{Unbalanced}
e^{i3\alpha^2 t \hat{P}_{k}\hat{X}_{j}^2 } = \: &e^{i2\alpha \hat{X}_{j} \hat{X}_{k}} e^{it\hat{P}_{k}^3}e^{-i\alpha \hat{X}_{j} \hat{X}_{k}}e^{-it\hat{P}_{k}^3}e^{-i2\alpha \hat{X}_{j} \hat{X}_{k}} e^{it\hat{P}_{k}^3}e^{i\alpha \hat{X}_{j} \hat{X}_{k}}e^{-it\hat{P}_{k}^3} e^{i\alpha^3 t \frac{3}{4}\hat{X}_{j}^3},
\end{align}
with $\alpha$ and $t$ real parameters. Before outlining the method in detail, we study simple examples to illustrate the main idea behind our approach.\\

Suppose that the goal is to derive an exact decomposition for the unitary $e^{i\alpha\hat{X}_{j}\hat{X}_{k}\hat{X}_{l}}$. The first step of the method is to express the operator $\hat{X}_{j}\hat{X}_{k}\hat{X}_{l}$ as a linear combination of polynomials of degree three in the quadrature operators $\hat{X}_{j}$, $\hat{X}_{k}$, and $\hat{X}_{l}$. Namely, we employ the identity
\begin{align}\label{ExThreeMode}
\hat{X}_{j}\hat{X}_{k}\hat{X}_{l}=&\tfrac{1}{6}[(\hat{X}_{j} + \hat{X}_{k} + \hat{X}_{l})^3-(\hat{X}_{j} + \hat{X}_{k})^3-(\hat{X}_{j} + \hat{X}_{l})^3-(\hat{X}_{k} + \hat{X}_{l})^3+\hat{X}_{j}^3+\hat{X}_{k}^3+\hat{X}_{l}^3],
\end{align}
which implies the identity
\begin{align} \label{DecompCube2}
e^{i\alpha\hat{X}_{j}\hat{X}_{k}\hat{X}_{l}} = \: &e^{\frac{i\alpha}{6}\left(\hat{X}_{j} + \hat{X}_{k} + \hat{X}_{l}\right)^3} e^{\frac{-i\alpha}{6}\left(\hat{X}_{j} + \hat{X}_{k}\right)^3}e^{\frac{-i\alpha}{6}\left(\hat{X}_{j} + \hat{X}_{l}\right)^3}e^{\frac{-i\alpha}{6}\left(\hat{X}_{k} + \hat{X}_{l}\right)^3}
e^{\frac{i\alpha}{6}\hat{X}^{3}_{j}} e^{\frac{i\alpha}{6}\hat{X}^{3}_{k}} e^{\frac{i\alpha}{6}\hat{X}^{3}_{l}},
\end{align}
since all the terms in the exponent commute. The right-hand side of this equation includes gates of the form $e^{\frac{i}{6}\hat{X}^{3}}$ that are part of the universal set, but it is still necessary to decompose the remaining terms. To do this, we employ the decompositions
\begin{align}\label{ThreeModePoly1}
e^{i\alpha\left(\hat{X}_{j} + \hat{X}_{k}\right)^3} &= e^{2i\hat{P}_{j}\hat{X}_{k}} e^{i\alpha\hat{X}_{j}^3} e^{-2i\hat{P}_{j}\hat{X}_{k}}\\
\label{ThreeModePoly2} e^{i\alpha\left(\hat{X}_{j} + \hat{X}_{k}+ \hat{X}_{l}\right)^3} &= e^{2i\hat{P}_{j}\hat{X}_{l}} e^{i\alpha(\hat{X}_{j}+\hat{X}_k)^3} e^{-2i\hat{P}_{j}\hat{X}_{l}},
\end{align}
which can be derived from Eqs.~\eqref{Eq:Unit_conj} and \eqref{Eq:BCH} using $U=e^{2i\hat{P}_{j}\hat{X}_{k}}$ as the unitary of conjugation. In summary, we have derived an exact decomposition by expressing $\hat{X}_{j}\hat{X}_{k}\hat{X}_{l}$ as a linear combination of polynomials of operators, allowing us to write the target gate $e^{i\alpha\hat{X}_{j}\hat{X}_{k}\hat{X}_{l}}$ in terms of a product of gates, each of which can be exactly decomposed. \\

Now suppose that the goal is to derive an exact decomposition for the higher-order single-mode gate $e^{i\alpha\hat{X}_j^4}$. Following our previous strategy, we seek to express the operator $\hat{X}_j^4$ as a linear combination of degree-four polynomials. It holds that 
\beq
\hat{X}_j^4 = (\hat{X}_{j}^2 + \hat{X}_{k})^2-\hat{X}_k^2-2\hat{X}_j^2\hat{X}_k,
\eeq
which leads to the identity
\beq \label{fourthorder}
e^{i\alpha\hat{X}_j^4}=e^{i\alpha(\hat{X}_{j}^2 + \hat{X}_{k})^2}e^{-i\alpha\hat{X}_k^2}e^{-2i\alpha\hat{X}_j^2\hat{X}_k}.
\eeq 
Here, the gate $e^{-i\alpha\hat{X}_k^2}$ is part of the universal set, while Eq.~\eqref{Unbalanced} gives an exact decomposition for $e^{-i\alpha\hat{X}_j^2\hat{X}_k}$ up to a Fourier transform. As before, the remaining term can be decomposed using unitary conjugation:
\begin{equation}\label{Eq:19}
e^{i\alpha(\hat{X}_{j}^2 + \hat{X}_{k})^2}=e^{2i\hat{P}_{k}\hat{X}_{j}^2} e^{i\alpha\hat{X}_{k}^2} e^{-2i\hat{P}_{k}\hat{X}_{j}^2},
\end{equation}
leading to a full decomposition for the target gate $e^{i \alpha \hat{X}_j^4}$. Note that an additional ancillary mode $k$ was required in this decomposition. To extend this method to a more general setting, we employ the same basic strategy: express the target gate in terms of a linear combination of polynomials and decompose the resulting gates in terms of unitary conjugation or previously derived decompositions.

\subsection{Single-mode gates}
We describe the method for decomposing single-mode gates of the form $e^{i\alpha\hat{X}^N}$ with $N$ an integer divisible by 2 or 3. In the previous example, we showed how Eq.~\eqref{Unbalanced} could be employed to decompose $e^{i\alpha\hat{X}^4}$. Generalizing Eq.~\eqref{Unbalanced} to higher order similarly enables decompositions of single-mode gates with larger exponents. It can be shown that such a general form exists, given by the expression
\begin{align}\label{twomode}
e^{2i\alpha^2 \hat{P}_{k}\hat{X}_{j}^{N}} = \: &e^{2i\alpha \hat{X}_{j}^{N-2} \hat{X}_{k}} e^{-i\alpha \hat{X}_{j}^2 \hat{P}_{k}^2} e^{-2i\alpha \hat{X}_{j}^{N-2} \hat{X}_{k}} e^{i\alpha \hat{X}_{j}^2 \hat{P}_{k}^2} e^{i\alpha^3 \hat{X}_{j}^{2(N-1)}},
\end{align}
for $N\geq2$. The proof of this formula can be found in the Appendix. This formula holds with the addition of another mode and can be proven in a similar manner.
\begin{align}\label{threemodeADD}
e^{2i\alpha^2 \hat{P}_{k} \hat{P}_{l}\hat{X}_{j}^n} = \: e^{2i\alpha \hat{X}_{j}^{n-2} \hat{X}_{k} \hat{X}_{l}} e^{-i\alpha \hat{X}_{j}^2 \hat{P}_{k}^2} e^{-2i\alpha \hat{X}_{j}^{n-2} \hat{X}_{k} \hat{X}_{l}} e^{i\alpha \hat{X}_{j}^2 \hat{P}_{k}^2} e^{-2i\alpha^3 \hat{X}_{j}^{2(n-1)} \hat{P}_{l}}.
\end{align}
These decompositions require the gate $e^{i\alpha \hat{X}^{2}_{j}\hat{X}^{2}_{k}}$, which is not part of the universal set. However, an exact decomposition also holds for this gate (see the Appendix for a proof):
\begin{align} \label{twosquares}
e^{i\alpha \hat{X}^{2}_{j}\hat{X}^{2}_{k}} = \: &e^{i2\hat{P}_{j}\hat{X}_{k}}e^{i\frac{\alpha}{12}\hat{X}_{j}^{4}}e^{-i4\hat{P}_{j}\hat{X}_{k}}e^{i\frac{\alpha}{12}\hat{X}_{j}^{4}}e^{i2\hat{P}_{j}\hat{X}_{k}} e^{-i\frac{\alpha}{6}\hat{X}_{j}^{4}}e^{-i\frac{\alpha}{6}\hat{X}_{k}^{4}},
\end{align}
where we can employ the previously derived decomposition for $e^{i\hat{X}_{j}^{4}}$. To obtain a general form for single-mode decompositions, we use Eq.~(\ref{twosquares}) as well as the fourth-order single-mode decomposition in Eq.~(\ref{fourthorder}) to first obtain a higher-order version of Eq.~(\ref{Unbalanced}):
\begin{align}
e^{2i\alpha^2 \hat{P}_{k}\hat{X}_{j}^3 } = \: & e^{2i\alpha \hat{X}_{j} \hat{X}_{k}} e^{-i\alpha \hat{X}_{j}^2 \hat{P}_{k}^2} e^{-2i\alpha \hat{X}_{j} \hat{X}_{k}} e^{i\alpha \hat{X}_{j}^2 \hat{P}_{k}^2} e^{-2i\alpha^3 \hat{X}_{j}^4}.
\end{align}
This can then be used to create a decomposition for $e^{i\hat{X}_j^6}$ in a similar way to the decomposition of the gate $e^{i\hat{X}_j^4}$. The decomposition for $e^{i\hat{X}_j^6}$ can once more be combined with Eq.~(\ref{twosquares}) to derive an exact decomposition for the next highest power of the two-mode gate, namely $e^{2i\alpha^2 \hat{P}_{k}\hat{X}_{j}^4}$. This process can be continued until the general recursive form in Eq.~(\ref{twomode}) is reached, as well as a more general decomposition of single-mode operations:
\begin{equation} \label{SingleGen}
e^{i\alpha\hat{X}_{k}^N} = e^{2i\hat{P}_{j}\hat{X}_{k}^{N/2}} e^{i\alpha\hat{X}_{j}^2} e^{-2i\hat{P}_{j}\hat{X}_{k}^{N/2}} e^{-i\alpha\hat{X}_{j}^2}  e^{-2i\alpha\hat{X}_{j}\hat{X}_{k}^{N/2}},
\end{equation}
that holds when $N$ is even. The proof of this equation is detailed in the Appendix, but follows similar steps to the fourth-order single-mode gate in Eq.~(\ref{fourthorder}). If $N$ is odd and a multiple of three, exact decompositions can also be derived by noting the following relation:
\begin{align}
 2\hat{X}_{k}^N=&2\left(\hat{X}_{j}+\hat{X}_{k}^{N/3}\right)^3 -3\left(\hat{X}_{l}+\hat{X}_{j}^{2} + \hat{X}_{k}^{N/3}\right)^2- 2\hat{X}_{j}^{3}+3\hat{X}_{j}^{4} + 3\hat{X}_{k}^{2N/3} - 6\hat{X}_{j}\hat{X}_{k}^{2N/3} + 6\hat{X}_{j}^2\hat{X}_{l} + 6\hat{X}_{k}^{N/3}\hat{X}_{l} + 3\hat{X}_{l}^{2}.
\end{align}
Therefore, for $N$ odd and divisible by 3, we can decompose the single-mode operation as
\begin{align} \label{OddSingle}
e^{i2\alpha\hat{X}_{k}^N} =  e^{i2\alpha\left(\hat{X}_{j}+\hat{X}_{k}^{N/3}\right)^3} e^{-i3\alpha\left(\hat{X}_{l}+\hat{X}_{j}^{2} + \hat{X}_{k}^{N/3}\right)^2} e^{-i2\alpha\hat{X}_{j}^{3}}e^{i3\alpha\hat{X}_{j}^{4}}e^{i3\alpha\hat{X}_{k}^{2N/3}} e^{-i6\alpha\hat{X}_{j}\hat{X}_{k}^{2N/3}}e^{i6\alpha\hat{X}_{j}^2\hat{X}_{l}} e^{i6\alpha\hat{X}_{k}^{N/3}\hat{X}_{l}} e^{i3\alpha\hat{X}_{l}^{2}}.
\end{align}
Here, the gates $e^{i2\alpha\left(\hat{X}_{j}+\hat{X}_{k}^{N/3}\right)^3}$ and $e^{-i3\alpha\left(\hat{X}_{l}+\hat{X}_{j}^{2} + \hat{X}_{k}^{N/3}\right)^2}$ can be decomposed using the expressions
\begin{align}
e^{i2\alpha\left(\hat{X}_{j}+\hat{X}_{k}^{N/3}\right)^3} &= e^{2i\hat{P}_{j}\hat{X}_{k}^{N/3}} e^{i2\alpha\hat{X}_{j}^3} e^{-2i\hat{P}_{j}\hat{X}_{k}^{N/3}}\label{Eq:27}, \\
e^{-i3\alpha\left(\hat{X}_{l}+\hat{X}_{j}^{2} + \hat{X}_{k}^{N/3}\right)^2} &= e^{2i\hat{P}_{l}\hat{X}_{k}^{N/3}} e^{2i\hat{P}_{l}\hat{X}_{j}^{2}} e^{-i3\alpha\hat{X}_{l}^2} e^{-2i\hat{P}_{l}\hat{X}_{j}^{2}}e^{-2i\hat{P}_{l}\hat{X}_{k}^{N/3}}\label{Eq:28},
\end{align}
which as before are obtained using unitary conjugation. The other gates in Eq.~(\ref{OddSingle}) can be decomposed with the previous general formulas Eq.~(\ref{SingleGen}) and Eq.~(\ref{twomode}).

\subsection{Multi-mode gates}
 
We study the case where $\hat{H}$ is given by
\beq \label{Hamiltonian}
\hat{H}=\prod_{j=1}^N \hat{X}_j^{n_j},
\eeq
where the $n_j$ are positive integers. We discuss later why restrictions are necessary on the exponents $n_j$, leading to exact decompositions for operators as in Eq.~\eqref{Eq: class_of_H}.\\

As discussed previously, the first step to decompose a multi-mode gate $e^{it \hat{H}}$ is to express $\hat{H}$ as a linear combination of operators. Let $[N]^k$ be the set of all $k$-subsets of $\{1,2,\ldots, N\}$, i.e., all subsets containing $k$ elements. For example, $[3]^2=\{\{1,2\}, \{1,3\}, \{2,3\}\}$. The goal is to find coefficients $c_1, c_2, \ldots, c_{N}$ such that \cite{kan2008moments}
\begin{align} \label{Eq: GeneralD}
\prod_{j=1}^N \hat{X}_j^{n_j}= \sum_{k=1}^{N}c_k\sum_{S\in[N]^k}\left(\sum_{i=1}^k\hat{X}_{S_i}^{n_{S_i}}\right)^N,
\end{align}
where $S\in[N]^k=\{S_1, S_2, \ldots, S_k\}$. When expanded, the term on the right-hand side contains several monomials of the position operators, including the desired term $\prod_{j=1}^N \hat{X}_j^{n_j}$. Each monomial is multiplied by a factor that is a linear combination of the coefficients $c_k$, and the goal is to set these factors to zero for all monomials except $\prod_{j=1}^N \hat{X}_j^{n_j}$. As shown in the Appendix, this gives rise to a linear system of equations for the coefficients $c_{k}$ such that Eq.~\eqref{Eq: GeneralD} holds whenever the coefficients $\vec{c}=(c_N,c_{N-1},\ldots, c_1)$ satisfy the linear system $A\vec{c}=0$. The matrix $A$ is independent of the exponents $n_j$ and is given by
\beq\label{Eq:PascalA}
A=\begin{pmatrix}
    1 & 1 & 0 & 0 & \dots  & 0 \\
    1 & 2 & 1 & 0 &\dots  & 0 \\
    1 & 3 & 3 & 1 &\dots  & 0 \\
    \vdots & \vdots & \vdots & \vdots & \ddots & \vdots \\
    \binom{N-1}{0} & \binom{N-1}{1} & \binom{N-1}{2} & \binom{N-1}{3} & \dots  & \binom{N-1}{N-1}
\end{pmatrix},
\eeq
i.e., the coefficients of $A$ follow the structure of Pascal's triangle.  Note that this linear system is underdetermined since there are $N-1$ equations for $N$ variables. However, by fixing $c_N$, it is possible to find a simple specific solution, as shown in the following observation.

\begin{obs}
A solution to the linear system $A\vec{c}=0$ with $\vec{c}=(c_N,c_{N-1},\ldots, c_1)$ and $A$ as in Eq.~\eqref{Eq:PascalA} is given by $c_{N-k}=(-1)^k c_N$. 
\end{obs}
\begin{proof}
For simplicity and without loss of generality, let $c_{N}=1$. The base case for $N=2$ is trivially true; it is simply $c_{2}+c_{1}=0 \implies c_1=-1$. Now examine the general structure for the case with $N=k$. Assume that the claimed solution $c_{N-k}=(-1)^k$  with $k=0,1,\ldots, N-2$ is true for $N=K-1$, i.e., the system when the last row and last column are omitted from the matrix $A$. For the case $N=K$, the last row of $A$ determines an equation for the remaining coefficient $c_1$. We then have
\begin{align}
&\sum_{k=0}^{K-1}\binom{K-1}{k}c_k =0\\
=&\sum_{\ell=0}^{K-1}\binom{K-1}{K-\ell}c_{K-\ell}\nonumber\\
=&c_1+\sum_{\ell=0}^{K-2}\binom{K-1}{K-\ell}(-1)^\ell=0.
\end{align}
We want to show that $c_1=(-1)^{K-1}$ is a solution to this equation. This yields
\begin{align}
(-1)^{K-1}+\sum_{\ell=0}^{K-2}\binom{K-1}{K-\ell}(-1)^\ell&=\sum_{\ell=0}^{K-1}\binom{K-1}{K-\ell}(-1)^\ell\nonumber\\
&=(-1+1)^{K-1}=0
\end{align}
as desired, where the last line follows from the binomial theorem.
\end{proof}
The solution $c_{N-k}=(-1)^k c_N$ is valid for any value of $c_N$. In order to satisfy Eq.~\eqref{Eq: GeneralD} exactly, we simply fix $c_N=1/N!$. With this choice of coefficients $c_k$, the sum of polynomials on the right-hand side of Eq.~(\ref{Eq: GeneralD}) is exactly equal to the multi-mode product of operators on the left-hand side. Thus, the process for decomposing multi-mode gates is to find an exact decomposition for each polynomial appearing on the right-hand side of Eq.~(\ref{Eq: GeneralD}). As done before, specifically in Eqs.~\eqref{ThreeModePoly1}, \eqref{ThreeModePoly2}, \eqref{Eq:19}, \eqref{Eq:27}, and \eqref{Eq:28}, decomposition of polynomials is performed using unitary conjugation with the gate $e^{2i\hat{P}_{1}\hat{X}_{j}^{n_j}}$  -- with decomposition in Eq.~(\ref{twomode}) -- and the lemma to the BCH formula. More precisely, we employ the following identity to decompose an arbitrary polynomial:
\begin{equation}\label{ArbPoly}
e^{it(\hat{X}_{1} + \hat{X}_{2}^{n_2} +  \hat{X}_{3}^{n_3} + \cdots + \hat{X}_{m}^{n_m})^N} = e^{2i\hat{P}_{1}\hat{X}_{m}^{n_m}}\cdots e^{2i\hat{P}_{1}\hat{X}_{3}^{n_3}}e^{2i\hat{P}_{1}\hat{X}_{2}^{n_2}} e^{it\hat{X}_{1}^N}e^{-2i\hat{P}_{1}\hat{X}_{2}^{n_2}}e^{-2i\hat{P}_{1}\hat{X}_{3}^{n_3}}\cdots e^{-2i\hat{P}_{1}\hat{X}_{m}^{n_m}}.
\end{equation}
Using Eq.~\eqref{ArbPoly}, it is not possible to find exact decompositions for all operators $\prod_{j=1}^N \hat{X}_j^{n_j}$ in Eq.~\eqref{Hamiltonian}, as there are restrictions on the exponents $n_j$. The restrictions are as follows:

\begin{enumerate}
\item There can exist at most one $j$ such that $n_j\neq1$. This restriction arises from the fact that the central operator in Eq.~\eqref{ArbPoly}, namely  $\hat{X}_1$, must have an exponent equal to one. Therefore, in order to use Eq.~\eqref{ArbPoly} to decompose every polynomial $\left(\sum_{i=1}^k\hat{X}_{S_i}^{n_{S_i}}\right)^N$ on the right-hand side of Eq.~(\ref{Eq: GeneralD}), all $k$-subsets $S$ with $k>1$ must contain at least one element $S_i\in S$ such that $n_{S_i}=1$. This is only possible if there exists at most one $j$ such that $n_j\neq1$.


\item The product $N n_j$ must be divisible by either 2 or 3 for all $j$. This arises because the $k=1$ terms in Eq.~(\ref{Eq: GeneralD}) produce monomials that include only single-mode operators to the power of $N n_j$. As shown in the previous section, the method only produces exact decompositions for single-mode operations with power divisible by 2 or 3.
\end{enumerate}

To summarize, we employ Eq.~\eqref{Eq: GeneralD} to express a multi-mode operator as a linear combination of polynomials. Each polynomial can then be exactly decomposed using Eq.~\eqref{ArbPoly} and single-mode decompositions from the previous section. This yields a method for constructing exact decompositions of operators of the form $e^{it\hat{H}}$, for $\hat{H}= \left(\prod_{j=1}^{N-1} \hat{X}_{j}\right)\hat{X}_N^n$, with both $Nn$ and $N$ divisible by either 2 or 3.

\section{APPLICATIONS AND IMPROVEMENTS} \label{Comparing}

In this section we demonstrate the power of exact decompositions by comparing our method to the standard commutator approximation in terms of circuit depth and precision. We also provide a table of CV algorithms and simulations of bosonic systems which contain operations that are covered by our method. The following table examines the gate counts for decompositions of some common operations. The gate counts neglect any Fourier transforms used by either method as they are inexpensive to implement experimentally.

\begin{center}
\begin{tabular}{ ||l|l|l|| } 
 \hline
 Target gate & Commutator approx. & Exact decomposition \\ 
 & ($10^{-3}$ precision) & \\
\hline \hline
 $e^{it\hat{X}^4}$ & $1.8\times 10^4$ gates & 29 gates \\ 
\hline
 $e^{it\hat{X}_{j}^2\hat{X}_{k}^2}$ & $2.8\times 10^4$ gates & 119 gates \\ 
 \hline
 $e^{it\hat{X}_{j}\hat{X}_{k}^3}$ & $2.9 \times 10^8$ gates & 125 gates \\ 
 \hline
$e^{it\hat{X}_{j}\hat{X}_{k}\hat{X}_{l}}$ & $4.2\times 10^8$ gates & 17 gates \\
\hline
$e^{it\hat{X}_{j}^{2}\hat{X}_{k}\hat{X}_{l}}$ & $1.4\times 10^9$ gates & 281 gates \\
\hline
$e^{it\hat{X}_{j}\hat{X}_{k}\hat{X}_{l}\hat{X}_{m}}$ & $6.9\times 10^{13}$ gates & 440 gates \\
\hline
\end{tabular}
\end{center}

The two gates with the lowest gate counts in the exact method are the third-order three-mode gate and the fourth-order single-mode gate, with 17 and 29 gates in their respective decompositions. This is in contrast to the commutator approximation where the addition of the third mode greatly increases the circuit depth. The structure of each method seems to indicate that the exact decompositions scale better under addition of more modes. Also, the need to repeat the set of gates to improve precision in the commutator approximation produces several orders-of-magnitudes increase in the resulting circuit depths. \\

To illustrate the applicability of these decompositions in quantum algorithms, the table below compiles a list of cases where the algorithm requires a decomposition for gates covered by our method. In some cases, only part of a desired operation might be decomposed exactly, but as demonstrated above, the exact decompositions even for these portions can produce a significant decrease in circuit depth.

\begin{center}
\begin{tabular}{ ||l|l|l|| } 
 \hline
 Algorithm & Hamiltonian and & Circuit depth \\ 
 & Operators covered by method & of operator \\
\hline \hline
 Vibrational dynamics of molecules, & $\hat{H} = \hbar \sum_{i\leq j}\frac{x_{ij}}{2}\sqrt{\omega_{i}\omega_{j}} \left(\hat{a}_{i}^{\dagger}\hat{a}_{i} + \hat{a}_{j}^{\dagger}\hat{a}_{j} + 2\hat{a}_{i}^{\dagger}\hat{a}_{j}^{\dagger}\hat{a}_{i}\hat{a}_{j}\right)$  & 119 gates \\ 
Ref.~\cite{sparrow2018simulating} & $e^{it\hat{H}}$ contains elements of universal set and operator $e^{it\hat{X}_{j}^2\hat{X}_{k}^2}$ & \\
\hline
 Non-homogeneous linear & $\hat{H} = \sum_{j=1}^N \left(a_j \hat{X}_j + b_j \hat{P}_j +\alpha_j\hat{X}_j^2+\beta_j\hat{P}_j^2\right)\hat{X}_k\hat{X}_l$  & 17 gates and \\ 
 partial differential equations, Ref.~\cite{CVHHL} & $e^{it\hat{H}}$ contains operators $e^{it\hat{X}_{j}\hat{X}_{k}\hat{X}_l}$ and $e^{it\hat{X}_{j}^2\hat{X}_{k}\hat{X}_l}$ & 281 gates \\
\hline
 Dipole interaction term of  & $\hat{H} = V\sum_{i\leq j}\hat{a}_i^{\dagger}\hat{a}_i\hat{a}_j^{\dagger}\hat{a}_j$  & 119 gates \\ 
Bose Hubbard model, Ref.~\cite{kalajdzievski2018continuous} & $e^{it\hat{H}}$ contains the operator $e^{it\hat{X}_{i}^2\hat{X}_{j}^2}$ & \\
\hline
 One particle tunneling term of  & $\hat{H} = -T\sum_{i\leq j}\hat{a}_i^{\dagger}\left(\hat{n}_i + \hat{n}_j\right)\hat{a}_j$  & 125 gates \\ 
Bose Hubbard model, Ref.~\cite{BHExtended} & $e^{it\hat{H}}$ contains the operator $e^{it\hat{X}_{i}\hat{X}_{j}^3}$ & \\
\hline
 Nearest-neighbor tunneling term of  & $\hat{H} = \frac{P}{2}\sum_{i\leq j}\hat{a}_i^{\dagger}\hat{a}_i^{\dagger}\hat{a}_j\hat{a}_j$  & 119 gates \\ 
Bose Hubbard model, Ref.~\cite{BHExtended} & $e^{it\hat{H}}$ contains the operator $e^{it\hat{X}_{i}^2\hat{X}_{j}^2}$ & \\
\hline
 Cross-Kerr Hamiltonian,  & $\hat{H} = \left(\hat{X}_i^2 + \hat{P}_i^2\right)\otimes \left(\hat{X}_j^2 + \hat{P}_j^2\right)$  & 119 gates \\ 
Ref.~\cite{CommutatorApprox} & $e^{it\hat{H}}$ contains the operator $e^{it\hat{X}_{i}^2\hat{X}_{j}^2}$ & \\
\hline
 Principal component analysis,  & $R(\hat{P}_R) = e^{i\delta\hat{P}_R(\hat{a}_1\hat{a}_2^{\dagger} + \hat{a}_1^{\dagger}\hat{a}_2)}$  & 17 gates \\ 
Ref.~\cite{ChrisAlgo2017} & $R(\hat{P}_R)$ contains the operator $e^{i\delta\hat{X}_{R}\hat{X}_{1}\hat{X}_2}$ & \\
\hline
 Matrix inversion algorithm,  & $R(\hat{P}_R\hat{P}_S) = e^{i\gamma\hat{P}_R\hat{P}_S(\hat{a}_1\hat{a}_2^{\dagger} + \hat{a}_1^{\dagger}\hat{a}_2)}$  & 440 gates \\ 
Ref.~\cite{ChrisAlgo2017} & $R(\hat{P}_R\hat{P}_S)$ contains the operator $e^{i\gamma\hat{X}_{R}\hat{X}_S\hat{X}_{1}\hat{X}_2}$ & \\
\hline
 Monte Carlo integration,  & $e^{ih(\hat{X}_1)\hat{P}_2\hat{P}_3\hat{P}_{\phi}}$, where $h(\hat{X}_1)$ is a polynomial in $\hat{X}_1$ & Depends on $h(\hat{X}_1)$ \\ 
Ref.~\cite{MonteCarlo} & can be decomposed by our method for any $h(\hat{X}_1)$ & \\
\hline
\end{tabular}
\end{center}

The final entry in the table contains a general operation that depends on the choice of $h(\hat{X}_1)$, which is chosen to be a polynomial in $\hat{X}_1$. This operation will be covered by the exact decomposition method regardless of the choice of $h(\hat{X}_1)$ because there are four total modes. Assuming $h(\hat{X}_1) = \hat{X}_1^{n}$, then $Nn=4n$ is always even and therefore the single mode operation $e^{it\hat{X}_1^{4n}}$ can be decomposed exactly. Also, since the final three modes are all to unit power, any one of them may be used as the exponent of the central operator in unitary conjugation as detailed in the previous section. Therefore both of the restrictions of the method have been met regardless of $n$. By linearity, the same holds for a general polynomial $h(\hat{X}_1) = \sum_{n}a_n\hat{X}_1^{n}$.

\section{CONCLUSION} \label{Discussion}
We presented a method for producing exact decompositions of continuous-variable operations into a product of gates from a universal set. In essence, the method works by expressing target Hamiltonians as a linear combination of polynomials, then finding exact decompositions of these polynomials using unitary conjugation in combination with the lemma to Baker-Campbell-Hausdorff. The unitaries covered by this method cover a large set of operations arising in photonic quantum algorithms and the simulation of bosonic system. Compared to previous techniques such as the standard commutator approximation, our method can yield reductions in gate count of several orders of magnitude, with the added advantage that the target unitaries are decomposed exactly. \\

Despite its wide applicability, our method does not produce exact decompositions for all possible bosonic gates. Notably, Hamiltonians that contain products of both $\hat{X}$ and $\hat{P}$ quadrature operators -- for instance operators of the form $\hat{H} = \hat{X}^n\hat{P}^m+\hat{P}^m\hat{X}^n$ -- are not covered by the method. Additionally, if the operator to be decomposed contains a sum of terms that do not commute, the Trotter-Suzuki approximation in Eq.~(\ref{Eq: trotter-suzuki}) still needs to be used to split the terms. An outstanding open question resulting from our work is to fully characterize the set of operations that can be decomposed exactly.\\

\section*{ACKNOWLEDGMENTS}
   We would like to thank Peter van Loock, Nicol\'as Quesada, Nathan Killoran, Sasho Kalajdzievski, and Christian Weedbrook for helpful discussions.

\onecolumngrid
\appendix
\section{Appendix}
\subsection{Proof of Eq.~\eqref{twomode}}
Here we prove the identity
\begin{align}\label{twomodeApp}
e^{2i\alpha^2 \hat{P}_{k}\hat{X}_{j}^{N}} = \: & e^{2i\alpha \hat{X}_{j}^{N-2} \hat{X}_{k}} e^{-i\alpha \hat{X}_{j}^2 \hat{P}_{k}^2} e^{-2i\alpha \hat{X}_{j}^{N-2} \hat{X}_{k}} e^{i\alpha \hat{X}_{j}^2 \hat{P}_{k}^2} e^{i\alpha^3 \hat{X}_{j}^{2(N-1)}}.
\end{align}

The middle three operators on the right-hand side can be expanded with unitary conjugation as
\begin{equation}
e^{-i\alpha \hat{X}_{j}^2 \hat{P}_{k}^2} e^{-2i\alpha \hat{X}_{j}^{N-2} \hat{X}_{k}} e^{i\alpha \hat{X}_{j}^2 \hat{P}_{k}^2} = e^{-2i\alpha \left(e^{-i\alpha \hat{X}_{j}^2 \hat{P}_{k}^2}\hat{X}_{j}^{N-2}e^{i\alpha \hat{X}_{j}^2 \hat{P}_{k}^2} \right)\left( e^{-i\alpha \hat{X}_{j}^2 \hat{P}_{k}^2}\hat{X}_{k}e^{i\alpha \hat{X}_{j}^2 \hat{P}_{k}^2}\right) }.
\end{equation}
Then using the lemma to BCH,  the two factors in the exponent can be simplified to get
\begin{equation}
e^{-i\alpha \hat{X}_{j}^2 \hat{P}_{k}^2}\hat{X}_{j}^{N-2}e^{i\alpha \hat{X}_{j}^2 \hat{P}_{k}^2} = \hat{X}_{j}^{N-2},
\end{equation}
and
\begin{equation}
e^{-i\alpha \hat{X}_{j}^2 \hat{P}_{k}^2}\hat{X}_{k}e^{i\alpha \hat{X}_{j}^2 \hat{P}_{k}^2} = \hat{X}_{k} - \alpha \hat{X}_j^2\hat{P}_k.
\end{equation}
The resulting two terms in the exponent are then separated using the Zassenhaus formula of Eq.~(\ref{splitting}):
\begin{align}
 e^{-2i\alpha \hat{X}_{j}^{N-2} \left(\hat{X}_{k}-\alpha \hat{X}_j^2\hat{P}_k\right)}& =  e^{-2i\alpha \hat{X}_{j}^{N-2} \hat{X}_{k}}e^{2i\alpha^2 \hat{X}_j^N\hat{P}_k}e^{-\frac{1}{2}[-2i\alpha \hat{X}_j^{N-2}\hat{P}_k, \: 2i\alpha^2\hat{X}_j^N\hat{P}_k]} \nonumber \\
&= e^{-2i\alpha \hat{X}_{j}^{N-2} \hat{X}_{k}}e^{2i\alpha^2 \hat{X}_j^N\hat{P}_k}e^{-i\alpha^3\hat{X}_j^{2(N-1)}}.
\end{align}
The two outside operators, $e^{-2i\alpha \hat{X}_{j}^{N-2} \hat{X}_{k}}$ and $e^{-i\alpha^3\hat{X}_j^{2(N-1)}}$ then cancel with the remaining two operators in Eq.~(\ref{twomode}), leaving the desired operator $e^{2i\alpha^2 \hat{X}_j^N\hat{P}_k}$.

\subsection{Proof of Eq.~(\ref{twosquares})}
Here we prove the following exact decomposition formula for the gate $e^{i\alpha \hat{X}^{2}_{j}\hat{X}^{2}_{k}}$:
\begin{align} \label{twosquaresAppendix}
e^{i\alpha \hat{X}^{2}_{j}\hat{X}^{2}_{k}} = \: &e^{i2\hat{P}_{j}\hat{X}_{k}}e^{i\frac{\alpha}{12}\hat{X}_{j}^{4}}e^{-i4\hat{P}_{j}\hat{X}_{k}}e^{i\frac{\alpha}{12}\hat{X}_{j}^{4}}e^{i2\hat{P}_{j}\hat{X}_{k}} e^{-i\frac{\alpha}{6}\hat{X}_{j}^{4}}e^{-i\frac{\alpha}{6}\hat{X}_{k}^{4}}.
\end{align}

We begin by expressing the operator $\hat{X}^{2}_{j}\hat{X}^{2}_{k}$ as a linear combination of polynomials:
\beq
\hat{X}^{2}_{j}\hat{X}^{2}_{k}=\tfrac{1}{12}\left(\hat{X}_j + \hat{X}_k \right)^4+\tfrac{1}{12}\left(\hat{X}_j - \hat{X}_k \right)^4-\tfrac{1}{6}\hat{X}_{j}^{4}-\tfrac{1}{6}\hat{X}_{k}^{4},
\eeq
which leads to the identity
\begin{equation}\label{Eq:AppendixUnConj1}
e^{i\alpha \hat{X}^{2}_{j}\hat{X}^{2}_{k}} = e^{i\frac{\alpha}{12}\left(\hat{X}_j + \hat{X}_k \right)^4}e^{i\frac{\alpha}{12}\left(\hat{X}_j - \hat{X}_k \right)^4}e^{-i\frac{\alpha}{6}\hat{X}_{j}^{4}}e^{-i\frac{\alpha}{6}\hat{X}_{k}^{4}}.
\end{equation}
Finally, from unitary conjugation it holds that 
\begin{equation}
 e^{i2\hat{P}_{j}\hat{X}_{k}}e^{i\frac{\alpha}{12}\hat{X}_{j}^{4}}e^{-i4\hat{P}_{j}\hat{X}_{k}}e^{i\frac{\alpha}{12}\hat{X}_{j}^{4}}e^{i2\hat{P}_{j}\hat{X}_{k}} = e^{i\frac{\alpha}{12}\left(\hat{X}_j + \hat{X}_k \right)^4}e^{i\frac{\alpha}{12}\left(\hat{X}_j - \hat{X}_k \right)^4},
\end{equation}
which gives Eq.~\eqref{twosquaresAppendix} when replaced in Eq.~\eqref{Eq:AppendixUnConj1}.

\subsection{Proof of Eq.~\eqref{SingleGen}}
Here we show the recursive decomposition for single-mode gates $e^{i\alpha\hat{X}_{k}^N}$:

\begin{equation} \label{SingleGenAppendix}
e^{i\alpha\hat{X}_{k}^N} = e^{2i\hat{P}_{j}\hat{X}_{k}^{N/2}} e^{i\alpha\hat{X}_{j}^2} e^{-2i\hat{P}_{j}\hat{X}_{k}^{N/2}} e^{-i\alpha\hat{X}_{j}^2}  e^{-2i\alpha\hat{X}_{j}\hat{X}_{k}^{N/2}}.
\end{equation}
As usual, we begin by expressing the target operator as a linear combination of polynomials:
\begin{equation}
\hat{X}^{N}_{k} = \left(\hat{X}_j + \hat{X}_k^{N/2} \right)^2-\hat{X}_{j}^2-\hat{X}_{j}\hat{X}_{k}^{N/2}
\end{equation}
which leads to the identity
\begin{equation}\label{Eq:IdenAppendix}
e^{i\alpha\hat{X}^{N}_{k}} = e^{i\alpha\left(\hat{X}_j + \hat{X}_k^{N/2} \right)^2}e^{-i\alpha\hat{X}_{j}^2}  e^{-2i\alpha\hat{X}_{j}\hat{X}_{k}^{N/2}}.
\end{equation}
From unitary conjugation it holds that 
\begin{equation}
e^{2i\hat{P}_{j}\hat{X}_{k}^{N/2}} e^{i\alpha\hat{X}_{j}^2} e^{-2i\hat{P}_{j}\hat{X}_{k}^{N/2}} = e^{i\alpha\left(\hat{X}_j + \hat{X}_k^{N/2} \right)^2},
\end{equation}
which leads to Eq.~\eqref{SingleGenAppendix} when replaced in Eq.~\eqref{Eq:IdenAppendix}.

\subsection{Derivation of the linear system of equations}
Here we show that finding coefficients $c_k$ such that the relation
\begin{align} \label{Eq: GeneralDApp}
\prod_{j=1}^N \hat{X}_j^{n_j}= \sum_{k=1}^{N}c_k\sum_{S\in[N]^k}\left(\sum_{i=1}^k\hat{X}_{S_i}^{n_{S_i}}\right)^N
\end{align}
holds is equivalent to solving the linear system $A\vec{c}=0$, with $\vec{c}=(c_N,c_{N-1},\ldots, c_1)$ and $A$ given by
\beq\label{Eq:PascalA_Appendix}
A=\begin{pmatrix}
    1 & 1 & 0 & 0 & \dots  & 0 \\
    1 & 2 & 1 & 0 &\dots  & 0 \\
    1 & 3 & 3 & 1 &\dots  & 0 \\
    \vdots & \vdots & \vdots & \vdots & \ddots & \vdots \\
    \binom{N-1}{0} & \binom{N-1}{1} & \binom{N-1}{2} & \binom{N-1}{3} & \dots  & \binom{N-1}{N-1}
\end{pmatrix}.
\eeq
Define $Y_j:=X_j^{n_j}$ such that Eq.~\eqref{Eq: GeneralDApp} becomes
\begin{align} \label{Eq: GeneralDAppYs}
\prod_{j=1}^N \hat{Y}_j= \sum_{k=1}^{N}c_k\sum_{S\in[N]^k}\left(\sum_{i=1}^k\hat{Y}_{S_i}\right)^N.
\end{align}
The expansion of the right-hand side of Eq.~\eqref{Eq: GeneralDAppYs} produces monomials of the form $\prod_{j=1}^N Y_j^{m_j}$, where $\sum_{j=1}^N m_j=N$ and the exponents $m_j\geq 0$ are non-negative integers. Each monomial can be uniquely labelled by the vector of exponents $\vec{m}=(m_1,m_2,\ldots,m_N)$. For each polynomial $\left(\sum_{i=1}^k\hat{Y}_{S_i}\right)^N$, it follows from the multinomial theorem that the coefficient in front of the monomial $\prod_{j=1}^N Y_j^{m_j}$ is always the same, namely the multinomial coefficient $\binom{N}{m_1,m_2,\ldots,m_N}$. For example, the polynomials $\left(Y_1 + Y_2 \right)^3$ and  $\left(Y_1 + Y_2 + Y_{3}\right)^3$ both produce a monomial $Y_1Y_2^2$ with coefficient $\binom{3}{1,1,0}=3$. Therefore, the overall coefficient $\chi_{\vec{m}}$ accompanying the monomial $\prod_{j=1}^N Y_j^{m_j}$ is given by
\beq
\chi_{\vec{m}}=\binom{N}{m_1,m_2,\ldots,m_N}\sum_{k=1}^Nc_kf_k(\vec{m}),
\eeq
where $f_k(\vec{m})$ is the number of times the monomial $\prod_{j=1}^N Y_j^{m_j}$ appears in polynomials of $k$ variables. The goal is to find coefficients $c_k$ such that $\chi_{\vec{m}}=0$ for all $\vec{m}$ except the target case $\vec{m}=(1,1,\ldots,1)$. This leads to the equations
\beq
\sum_{k=1}^Nc_kf_k(\vec{m})=0.
\eeq
Suppose that $\vec{m}$ has $\ell$ non-zero elements, i.e., the monomial $\prod_{j=1}^N Y_j^{m_j}$ contains $\ell$ variables. The quantity $f_k(\vec{m})$ is then equal to the number of ways in which the remaining $k-\ell$ variables can be selected from the remaining $N-\ell$ ones, which is simply $\binom{N-\ell}{k-\ell}$. Thus, $f_k(\vec{m})=\binom{N-\ell}{k-\ell}$, which depends only on $\ell$, leading to $N-1$ equations for each $\ell=1,2,\ldots, N-1$:
\beq
\sum_{k=1}^Nc_k\binom{N-\ell}{k-\ell}=:A\vec{c}=0,
\eeq
with $A$ as in Eq.~\eqref{Eq:PascalA_Appendix}.

\bibliographystyle{apsrev}
\bibliography{DP}

\end{document}